\RequirePackage{amsmath}
\documentclass[runningheads]{llncs}
\usepackage{amssymb}
\usepackage{verbatim}
\usepackage{gastex}
\usepackage{algorithmicx,algorithm,algpseudocode}
\usepackage{pstricks,pstricks-add,pst-node}
\usepackage{array,graphicx}
\usepackage{xcolor}
\usepackage[lowtilde]{url}

\usepackage[utf8]{inputenc}
\usepackage[T1]{fontenc}
\DeclareSymbolFont{rsfscript}{OMS}{rsfs}{m}{n}
\DeclareSymbolFontAlphabet{\mathrsfs}{rsfscript}

\renewcommand{\O}{\mathcal{O}}

\title{Checking Whether an Automaton Is Monotonic\\Is NP-complete}
\author{Marek Szyku{\l}a\thanks{Supported in part by Polish NCN grant DEC-2013/09/N/ST6/01194.}}  
\institute{
Institute of Computer Science, University of Wroc{\l}aw,\\
Joliot-Curie 15, PL-50-383 Wroc{\l}aw, Poland\\
\{{\tt msz@cs.uni.wroc.pl}\}}

\begin{document}
\maketitle
\begin{abstract}
An automaton is monotonic if its states can be arranged in a linear order that is preserved by the action of every letter. We prove that the problem of deciding whether a given automaton is monotonic is NP-complete.
The same result is obtained for oriented automata, whose states can be arranged in a cyclic order.
Moreover, both problems remain hard under the restriction to binary input alphabets.
\medskip

\end{abstract}


\section{Introduction}

We deal with complete deterministic finite (semi)automata $\mathcal{A} = \langle Q,\Sigma,\delta \rangle$, where $Q$ is the set of states, $\Sigma$ is the input alphabet, and $\delta\colon Q \times \Sigma \to Q$ is the transition function defining the action of $\Sigma$ on $Q$. This action naturally extends to the action of $\delta(q,w)$ words for any $q \in Q$, $w\in \Sigma^*$.

Monotonic automata are those that admit a linear order of the states. The same qualification is applied to transformation semigroups. Formally, an automaton $\mathcal{A}$ is \emph{monotonic} if there exists a linear order $\le$ of $Q$ such that if $p \le q$ then $\delta(p,a) \le \delta(q,a)$, for all $p,q \in Q$ and $a \in \Sigma$. We call such an order $\le$ an \emph{underlying linear order} of $\mathcal{A}$. It is clear that if the actions of all letters preserve the order, then also the actions of all words do so.

The class of monotonic automata is a subclass of aperiodic ones \cite{MP1971CounterFreeAutomata}, which recognize precisely \emph{star-free} languages, and form one of the fundamental classes in the theory of formal languages. An automaton is \emph{aperiodic} if no transformation of any word has a nontrivial cycle. Checking whether an automaton is aperiodic is known to be PSPACE-complete \cite{ChoHuynh1991}. On the other hand, checking whether an automaton is \emph{nonpermutational}, where no transformation acts like a permutation of a nontrivial subset of $Q$, can be easily done in $\O(|\Sigma| \times |Q|^2)$ time \cite{IvNG2014}. Such results may be useful in improving algorithms recognizing star-free languages to work better in particular cases.
The complexity problems for various subclasses of regular languages are widely studied (see~\cite{BSX2011DecisionProblemsForConvexLanguages} for regognizing convex, and~\cite{KMM1991APolynomialTimeAlgorithmForLocalTestabilityProblem} for locally testable languages, and~\cite{HolzerKutrib2011Survey} for a survey).
The languages of monotonic automata do not have bounded level in the \emph{dot-depth hierarchy} of star-free languages \cite{BrzozowskiKnast1978TheDotDepthHierarchyIsInfinite}.

Monotonic semigroups were studied by Gomes and Howie \cite{GomesHowie1992} for their maximum size (they use the term \emph{order-preserving}). These semigroups play an important role as building-blocks in the constructions of the largest aperiodic semigroups known so far (\cite{BLL13SyntacticComplexitiesOfSixClasses,BS2014LargeAperiodicSemigroups}).

Monotonic automata have been considered, in particular, in connection with the problems of synchronizing automata. An automaton is said to be \emph{synchronizing} if there is a word $w$ such that $|Qw|=1$; such a word is called a \emph{reset word}. The \v{C}ern\'y conjecture, which is considered one of the most longstanding open problem in automata theory, states that every synchronizing automaton has a reset word of length at most $(|Q|-1)^2$.
Ananichev and Volkov \cite{AV2003SynchronizingMonotonicAutomata} have proved that a synchronizing monotonic automaton has a reset word of length at most $|Q|-1$. They have also proved the same bound for a larger class of \emph{generalized monotonic} automata \cite{AV2005SynchronizingGeneralizedMonotonicAutomata}. Volkov have introduced a still larger class of \emph{weakly monotonic} automata \cite{Volkov2009ChainOfPartialOrders}, which contains all aperiodic ones, and proved that strongly connected automata in this class possess a synchronizing word of length $|Q|(|Q|+1)/6$. Finally, Grech and Kisielewicz have generalized this to the class of automata \emph{respecting intervals of a directed graph}, and they have proved that the \v{C}ern\'y conjecture holds for each automaton in this class, provided it holds for smaller \emph{quotient} automata.
These results could be also useful in computational verification of the conjecture for automata of limited size, provided we could efficiently recognize and skip from computations automata that belong to a class for which the conjecture has been proven \cite{KS2013GeneratingSmallAutomata,KS2014SynchronizingAutomataWithLargeResetLengths}. Therefore it is important to consider computational complexity of the related problems.

The term \emph{monotonic} was also used by Eppstein \cite{Ep1990} for automata whose states can be arranged in a \emph{cyclic order} that is preserved by the actions of the letters. Following \cite{AV2003SynchronizingMonotonicAutomata} we call such automata \emph{oriented} automata. They form a broader class, containing monotonic automata, which has certain applications in robotics (\emph{part-orienters}, see Natarajan \cite{Na1986}). Eppstein has established the tight upper bound for the length of the shortest reset words of an oriented automaton $(|Q|-1)^2$, and provided an algorithm working in $\O(|\Sigma| \times |Q|^2)$ time for finding such a word. However, this algorithm requires the cyclic order to be given.

Note that the problem of finding the shortest reset word is hard in general \cite{OM2010} (also for approximation \cite{Berlinkov2010Approximating,GH2011} and some restricted classes \cite{Ma2009}). But due to possible practical applications, there are many exponential algorithms that can deal with fairly large automata and polynomial heuristics (e.g.~\cite{KRW2012,KKS2015ComputingTheShortestResetWords,RSP1993,Sandberg2005Survey,ST2011}). Also, hardness does not exclude a possibility of using a polynomial algorithm for some easily tractable classes (cf.~slowly synchronizing \cite{KKS2015ComputingTheShortestResetWords}).

Here we prove that the problem of checking whether a given automaton is monotonic is NP-complete, even under restriction to binary alphabets (Section~\ref{sec:monotonic}). We also obtain that checking whether an automaton is oriented is NP-complete under the same conditions (Section~\ref{sec:oriented}).
It follows that, unfortunately, they are hardly recognizable, and it is hard to find a preserved linear (cyclic) order of a monotonic (oriented) automaton. In particular, we cannot efficiently apply the polynomial Eppstein algorithm \cite{Ep1990} to compute a shortest reset word in the cases oriented automata, without knowing a cyclic order. On the other hand, checking whether an automaton admits a nontrivial \emph{partial order} is easy (Section~\ref{sec:discussion}).


\section{Monotonic Automata}\label{sec:monotonic}

The problem $\textsc{MONOTONIC}$ can be formulated as follows: given an automaton $\mathcal{A}$, decide if $\mathcal{A}$ is monotonic. This is the unrestricted version, where the alphabet can be arbitrary large. For a given $k \ge 1$, the restricted problem to $k$-letter alphabets of the input automaton we call $\textsc{MONOTONIC}_k$.

We show that $\textsc{MONOTONIC}$ is NP-complete, as well as $\textsc{MONOTONIC}_k$ for any $k \ge 2$. The problem is easy if the alphabet is unary.

\begin{proposition}
A unary automaton is monotonic if and only if the transformation of the single letter does not contain a cycle of length $\ge 2$. $\textsc{MONOTONIC}_1$ can be solved in $\O(|Q|)$ time, and a monotonic order can be found in $\O(|Q|)$ time if it exists.
\end{proposition}
\begin{proof}
We simply check if the transformation of the single letter of $\mathcal{A}$ contains a cycle of length $\ge 2$, that is $\delta(q_1,a)=q_2,\delta(q_2,a)=q_3,\ldots,\delta(q_\ell,a)=q_1$ for some distinct states $q_1,\ldots,q_\ell$. If so, then from $q_1 < q_2$ (or dually $q_1 > q_2$) it follows that $q_2 < q_3,\ldots,q_\ell<q_1$---a contradiction with that $<$ is an order. Thus the automaton is not monotonic.

Otherwise we have an acyclic digraph of the transformation, and we can fix some order on the connected components (sometimes called \emph{clusters}). Each such a component form a rooted tree. We can perform an inverse depth-first search (DFS) starting from the root. Then $p \le q$ if $p$ is in a component before that of $q$, or they are in the same component but $p$ was visited later than $q$ during the inverse DFS in this component. So if $p \le q$ from the same component, then $\delta(p,a)$ was visited later than $\delta(q,a)$, or $\delta(p,a)=\delta(q,a)$. Thus the order is preserved.
These operations can be done in $\O(|Q|)$ time.
\qed
\end{proof}

Clearly, $\textsc{MONOTONIC}$ is in NP, as we can guess an underlying linear order and check if the action of each letter preserves it (this can be done in $\O(|\Sigma| \times |Q|)$ time).

\begin{proposition}
$\textsc{MONOTONIC}$ is in NP.
\end{proposition}

\subsection{$\textsc{MONOTONIC}$ is NP-complete}\label{subsec:monotonic_npc}

We reduce $\textsc{MONOTONE-NAE-3SAT}$ to $\textsc{MONOTONIC}$.

$\textsc{NAE-3SAT}$ ($\textsc{NOT-ALL-EQUAL}$) is a variant of 3SAT, where a clause is satisfied if it contains at least one true and one false literal. The variant $\textsc{MONOTONE-NAE-3SAT}$ additionally restricts instances so that every literal is a positive occurrence of a variable (negations are not allowed). From Schaefer's Theorem \cite{Schaefer1978}, we have that $\textsc{NAE-3SAT}$ is NP-complete as well as $\textsc{MONOTONE-NAE-3SAT}$.

As an instance $I$ of $\textsc{MONOTONE-NAE-3SAT}$ we get a set of $n$ boolean variables $\mathcal{V}=\{v_1,\ldots,v_n\}$, and a set of $m$ clauses $\mathcal{C}=\{C_1,\ldots,C_m\}$, each one with exactly 3 literals. A literal is a positive occurrence of a variable $v_i$. The problem is to decide if there exists a satisfying assignment $\sigma\colon \mathcal{V} \to \{0,1\}$ for $I$, that is, for each clause $C_i \in \mathcal{C}$, $C_i$ contains at least one true literal ($v_j \in C_i$ with $\sigma(v_j)=1$) and at least one false literal ($v_j \in C_i$ with $\sigma(v_j)=0$).
We can assume that each variable occurs at least one time, and no variable appears more than once in a clause.
Note that the complement of a satisfying assignment for $I$ is also satisfying.

\subsubsection{Definition of $\mathcal{A}_I$.}

We construct $\mathcal{A}_I = \langle Q,\Sigma,\delta\rangle$ as follows.
For each variable $v_i \in \mathcal{V}$ we create a pair of states $p_i,q_i$. We also add a unique state $s$ (sink).

For a $j$-th clause $C_j = (v_f,v_g,v_h)$ (we fix the order of variables in clauses), we create the \emph{clause gadget} as follows. We add three states $x_j,y_j,z_j$ and three letters $a_j,b_j,c_j$, which correspond to the three occurrences of the variables $v_f,v_g,v_h$, respectively. The action of these letters is defined as follows:
\begin{itemize}
\item $\delta(p_f,a_j) = x_j$ and $\delta(q_f,a_j) = y_j$;
\item $\delta(p_g,b_j) = y_j$ and $\delta(q_g,b_j) = z_j$;
\item $\delta(p_h,c_j) = z_j$ and $\delta(q_h,c_j) = x_j$;
\item $\delta(p_i,a_j) = p_i$ and $\delta(q_i,a_j) = q_i$, for $i < f$;
\item $\delta(p_i,b_j) = p_i$ and $\delta(q_i,b_j) = q_i$, for $i < g$;
\item $\delta(p_i,c_j) = p_i$ and $\delta(q_i,c_j) = q_i$, for $i < h$;
\item $\delta(u,e) = s$, for the other states $u$ and each $e \in \{a_j,b_j,c_j\}$.
\end{itemize}
So the actions of letters $a_j$, $b_j$, $c_j$ send every state from $Q \setminus \{p_i,q_i\}$ either to itself or to $s$.
The clause gadget is presented in~Figure~\ref{fig:clause_gadget}.

\begin{figure}[ht]
\unitlength 12pt
\begin{center}\begin{picture}(12,11)(0,0.5)
\gasset{Nh=2,Nw=2,Nmr=1,ELdist=0.5,loopdiam=1}
\node(x)(6,8){$x_j$}
\node(y)(4,4){$y_j$}
\node(z)(8,4){$z_j$}
\node(p_f)(2,10){$p_f$}
\node(q_f)(0,6){$q_f$}
\node(p_g)(4,0){$p_g$}
\node(q_g)(8,0){$q_g$}
\node(p_h)(12,6){$p_h$}
\node(q_h)(10,10){$q_h$}
\drawedge(p_f,x){$a_j$}
\drawedge(q_f,y){$a_j$}
\drawedge(p_g,y){$b_j$}
\drawedge(q_g,z){$b_j$}
\drawedge(p_h,z){$c_j$}
\drawedge(q_h,x){$c_j$}
\end{picture}\end{center}
\caption{The clause gadget for a $j$-th clause $(v_f,v_g,v_h)$.}
\label{fig:clause_gadget}
\end{figure}
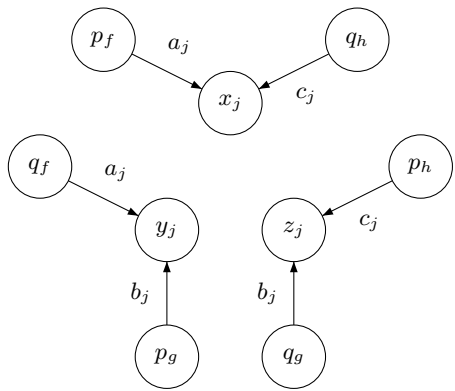

In~Figure~\ref{fig:construction_of_A_I} the construction of $\mathcal{A}_I$ is presented, with the action of $a_1$ as an example, in the case when variable $v_2$ is the first literal in clause $C_1$.

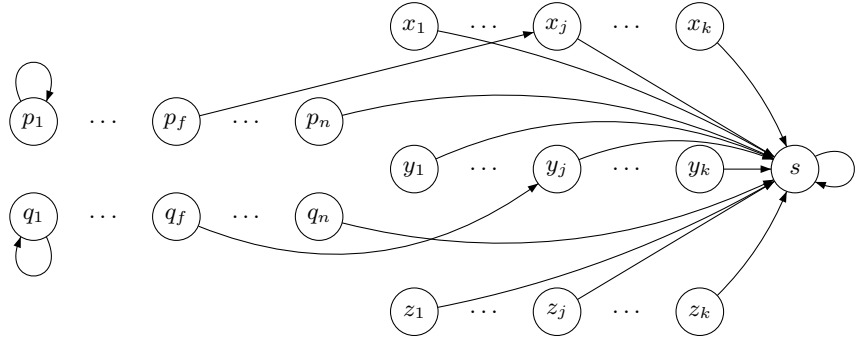
\begin{figure}[ht]
\unitlength 9pt
\begin{center}\begin{picture}(32,13)(0,0.5)
\gasset{Nh=2,Nw=2,Nmr=1,ELdist=0.5,loopdiam=1.5}
\node(p_1)(0,8){$p_1$}\node(q_1)(0,4){$q_1$}
\node[Nframe=n](p_dots1)(3,8){$\dots$}\node[Nframe=n](q_dots1)(3,4){$\dots$}
\node(p_f)(6,8){$p_f$}\node(q_f)(6,4){$q_f$}
\node[Nframe=n](p_dots2)(9,8){$\dots$}\node[Nframe=n](q_dots2)(9,4){$\dots$}
\node(p_n)(12,8){$p_n$}\node(q_n)(12,4){$q_n$}

\node(x_1)(16,12){$x_1$}\node(y_1)(16,6){$y_1$}\node(z_1)(16,0){$z_1$}
\node[Nframe=n](x_dots1)(19,12){$\dots$}\node[Nframe=n](y_dots1)(19,6){$\dots$}\node[Nframe=n](z_dots1)(19,0){$\dots$}
\node(x_j)(22,12){$x_j$}\node(y_j)(22,6){$y_j$}\node(z_j)(22,0){$z_j$}
\node[Nframe=n](x_dots2)(25,12){$\dots$}\node[Nframe=n](y_dots2)(25,6){$\dots$}\node[Nframe=n](z_dots2)(25,0){$\dots$}
\node(x_k)(28,12){$x_k$}\node(y_k)(28,6){$y_k$}\node(z_k)(28,0){$z_k$}
\node(s)(32,6){$s$}

\drawloop[loopangle=90](p_1){}\drawloop[loopangle=270](q_1){}
\drawedge[curvedepth=0](p_f,x_j){}\drawedge[curvedepth=-2.5](q_f,y_j){}
\drawedge[curvedepth=2](p_n,s){}\drawedge[curvedepth=-2](q_n,s){}
\drawedge[curvedepth=.8](x_1,s){}\drawedge[curvedepth=2](y_1,s){}\drawedge[curvedepth=-.8](z_1,s){}
\drawedge[curvedepth=0](x_j,s){}\drawedge[curvedepth=1.2](y_j,s){}\drawedge[curvedepth=0](z_j,s){}
\drawedge[curvedepth=.5](x_k,s){}\drawedge(y_k,s){}\drawedge[curvedepth=-.5](z_k,s){}
\drawloop[loopangle=0](s){}
\end{picture}\end{center}
\caption{The action of the letter $a_j$, where $v_f$ is the first variable in $C_j$.}
\label{fig:construction_of_A_I}
\end{figure}

In summary, we have $|Q|=2n+3m+1$ states and $|\Sigma|=3m$ letters.

\subsubsection{Correctness of the Reduction.}

\begin{theorem}\label{thm:monotonic_hard}
$\mathcal{A}_I$ is monotonic if and only if $I$ has a satisfying assignment.
\end{theorem}
\begin{proof}
Suppose that $\mathcal{A}_I$ is monotonic with the underlying linear order $\le$. We define an assignment $\sigma$ for $I$:
$\sigma(v_i) = 0$ if $p_i < q_i$, and $\sigma(v_i) = 1$ otherwise. We show that $\sigma$ is satisfying for $I$.

Assume for the contrary that there is a clause $C_j = (v_f,v_g,v_h)$, where all the three variables evaluate to $0$. This means that $p_f < q_f$, $p_g < q_g$, and $p_h < q_h$. From that $\le$ is preserved, we have:
\begin{itemize}
\item $\delta(p_f,a_j) = x_j < y_j = \delta(q_f,a_j)$;
\item $\delta(p_g,b_j) = y_j < z_j = \delta(q_g,b_j)$;
\item $\delta(p_h,c_j) = z_j < x_j = \delta(q_h,c_j)$.
\end{itemize}
Thus $x_j < y_j < z_j < x_j$, a contradiction with that $\le$ is an order.
The argument holds in the dual way in the case with all the three variables evaluated to $1$. Hence, $\sigma$ must be satisfying.

Now, suppose that there is a satisfying assignment $\sigma$. We define a linear order $\le$ and show that it is preserved.
To do so, we define $\tau\colon Q \to \mathbb{N}$, which for states $q \in Q$ assigns pairwise distinct natural numbers that will determine $\le$.

First, for any $1 \le i \le n$ let:
\begin{itemize}
\item $\tau(p_i) = 2i-1$ and $\tau(q_i) = 2i$ if $\sigma(v_i)=0$;
\item $\tau(p_i) = 2i$ and $\tau(q_i) = 2i-1$ if $\sigma(v_i)=1$.
\end{itemize}

For $u \in \{x_j,y_j,z_j\}$ we define $\tau(u) \in \{2n+3j-2,2n+3j-1,2n+3j\}$, depending on the assignment of the variables in $C_j = (v_f,v_g,v_h)$. Assignment $\sigma$ uniquely determines the relation between $x_j,y_j,z_j$ in an underlying linear order. Each of the six satisfying combinations of $\sigma(v_f),\sigma(v_g),\sigma(v_h)$ defines an acyclic relation between $x_j,y_j,z_j$, which is enforced by the action of the letters $a_j,b_j,c_j$.
For instance, if $\sigma(v_f)=0$, then $p_f < q_f$, which implies $\delta(p_f,a_j) = x_j < y_j = \delta(q_f,a_j)$. If $\sigma(v_g)=0$ then $y_j < z_j$. Then it must be $\sigma(v_h)=1$ and $z_j > x_j$. If $\sigma(v_g)=1$ then $y_j > z_j$, and we have either $z_j < x_j < y_j$ if $\sigma(v_h)=0$, or $x_j < z_j < y_j$ otherwise. This is dual for $\sigma(v_f)=1$.

Finally we define $\tau(s) = 3n+3m+1$. Hence, in our order $\le$, first there are states $p_i,q_i$ sorted increasingly by $i$. The order between $p_i$ and $q_i$ depends on the assignment. Next, there are states from clause gadgets $x_j,y_j,z_j$ sorted by $j$. The exact order on particular $x_j,y_j,z_j$ depends on the assignment as described above. Finally $s$ is the last state with $u \le s$ for any $u \in Q$. The order is shown in~Figure~\ref{fig:construction_of_A_I} (from left to right).

Now we show that $\le$ is indeed an underlying linear order. Consider a letter $a_j$ for any $1 \le j \le k$, and let $C_j=(v_f,v_g,v_h)$.
We show that for every pair of distinct states the order $\le$ is preserved.
\begin{itemize}
\item For the pair $p_f,q_f$, if $p_f < q_f$ then also $\delta(p_f,a_j) = x_j < y_j = \delta(q_f,a_j)$, and if $p_f > q_f$ then $\delta(p_f,a_j) = x_j > y_j = \delta(q_i,a_j)$, since we have chosen the order of $x_j,y_j,z_j$ to be consistent with $\sigma$, as described above.
\item For $p_f$ (or $q_f$) and $u \in Q \setminus \{p_f,q_f\}$, if $p_f < u$ then $\delta(p_f,a_j) = x_j < \delta(u,a_j) = s$. If $u < p_f$ then $\delta(u,a_j) = u < x_j = \delta(p_i,a^j_i)$. The same holds for $q_f$ mapped to $y_j$.
\item For distinct states $u,v \in Q \setminus \{p_f,q_f\}$ with $u < v$, if $\delta(u,a_j) = u$, then either $\delta(v,a_j) = v > u$ or $\delta(v,a_j) = s > u$. If $\delta(u,a_j) = s$ then also $\delta(v,a_j) = s$.
\end{itemize}
The same arguments work for letters $b_j$ and $c_j$. It follows that any letter preserves $\le$, so $\le$ is an underlying linear order of $\mathcal{A}_I$.
\qed
\end{proof}
 
We can state our main
\begin{theorem}
The problem of checking whether a given automaton is monotonic is NP-complete. 
\end{theorem}

\subsection{Reduction from $\textsc{MONOTONIC}$ to $\textsc{MONOTONIC}_2$}

Let $\mathcal{A}=\langle Q,\Sigma,\delta\rangle$ be an automaton with $Q = \{v_1,\ldots,v_n\}$ and $\Sigma = \{a_1,\ldots,a_k\}$ with $k \ge 3$.
We construct a binary automaton $\mathcal{B}_\mathcal{A}=\langle Q_\mathcal{B},\{a,b\},\delta_\mathcal{B}\rangle$ such that $\mathcal{A}$ is monotonic if and only if $\mathcal{B}_\mathcal{A}$ is monotonic.

$Q_\mathcal{B}$ consists of $kn$ states $q^i_j$ for $1 \le i \le k,1 \le j \le n$, and a unique state $s$ (sink).
Now we define the action of $a$. For each state $q^i_j$ with $1 \le i \le k-1$ and $1 \le j \le n$, we define $\delta_\mathcal{B}(q^i_j,a) = \delta_\mathcal{B}(q^{i+1}_j)$. For each $q^k_j$ we define $\delta_\mathcal{B}(q^k_j,a) = s$. Finally $\delta(s,a) = s$.
The action of $b$ in each set $\{q^i_1,\ldots,q^i_n\}$ corresponds to the action of the $i$-th letter of $\Sigma$ on $Q$: For $1 \le i \le k$ and $1 \le j \le n$, if $\delta(v_j,a_i) = v_g$ then we define $\delta_\mathcal{B}(q^i_j,b) = q^i_g$. Finally $\delta(s,b) = s$.
The construction of $\mathcal{B}_\mathcal{A}$ is shown in~Figure~\ref{fig:construction_of_B}.

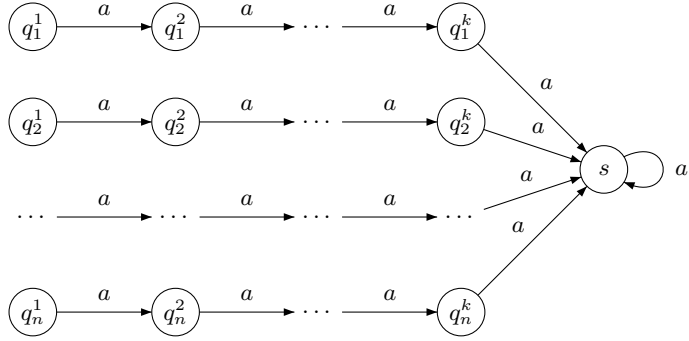
\begin{figure}[ht]
\unitlength 9pt
\begin{center}\begin{picture}(24,13)(0,0.5)
\gasset{Nh=2,Nw=2,Nmr=1,ELdist=0.5,loopdiam=1.5}
\node(q^1_1)(0,12){$q^1_1$}
\node(q^1_2)(0,8){$q^1_2$}
\node[Nframe=n](q^1_dots)(0,4){$\dots$}
\node(q^1_n)(0,0){$q^1_n$}
\node(q^2_1)(6,12){$q^2_1$}
\node(q^2_2)(6,8){$q^2_2$}
\node[Nframe=n](q^2_dots)(6,4){$\dots$}
\node(q^2_n)(6,0){$q^2_n$}
\node[Nframe=n](dots1)(12,12){$\dots$}
\node[Nframe=n](dots2)(12,8){$\dots$}
\node[Nframe=n](dotsdots)(12,4){$\dots$}
\node[Nframe=n](dotsn)(12,0){$\dots$}
\node(q^k_1)(18,12){$q^k_1$}
\node(q^k_2)(18,8){$q^k_2$}
\node[Nframe=n](q^k_dots)(18,4){$\dots$}
\node(q^k_n)(18,0){$q^k_n$}
\node(s)(24,6){$s$}
\drawedge(q^1_1,q^2_1){$a$}
\drawedge(q^1_2,q^2_2){$a$}
\drawedge(q^1_dots,q^2_dots){$a$}
\drawedge(q^1_n,q^2_n){$a$}
\drawedge(q^2_1,dots1){$a$}
\drawedge(q^2_2,dots2){$a$}
\drawedge(q^2_dots,dotsdots){$a$}
\drawedge(q^2_n,dotsn){$a$}
\drawedge(dots1,q^k_1){$a$}
\drawedge(dots2,q^k_2){$a$}
\drawedge(dotsdots,q^k_dots){$a$}
\drawedge(dotsn,q^k_n){$a$}
\drawedge(q^k_1,s){$a$}
\drawedge(q^k_2,s){$a$}
\drawedge(q^k_dots,s){$a$}
\drawedge(q^k_n,s){$a$}
\drawloop[loopangle=0](s){$a$}
\end{picture}\end{center}
\caption{The action of $a$ in $\mathcal{B}_\mathcal{A}$.}
\label{fig:construction_of_B}
\end{figure}

\begin{theorem}\label{thm:monotonic2_hard}
$\mathcal{B}_\mathcal{A}$ is monotonic if and only if $\mathcal{A}$ is monotonic.
\end{theorem}
\begin{proof}
Suppose that $\mathcal{A}$ is monotonic with the underlying linear order $\le_\mathcal{A}$. We define the linear order $\le_\mathcal{B}$ on the states of $\mathcal{B}_\mathcal{A}$.
For $1 \le i,f \le k$ and $1 \le j,g \le n$, let $q^i_j \le_\mathcal{B} q^f_g$ if and only if $i < f$, or $i=f$ and $v_j <_\mathcal{A} v_g$. Also, let $q^i_j <_\mathcal{B} s$ for each $i,j$. The order $\le_\mathcal{B}$ is linear, since $\le_\mathcal{A}$ is linear. We show that $\le_\mathcal{B}$ is an underlying linear order of $\mathcal{B}_\mathcal{A}$.

Clearly, the actions of both letters preserve $\le_\mathcal{B}$ on states $q^i_j$ and $s$.
Consider a pair $q^i_j,q^f_g$ with $q^i_j \le_\mathcal{B} q^f_g$. Then $i \le f$ by definition. Consider the following cases:
\begin{itemize}
\item If $i < f$, then $\delta_\mathcal{B}(q^i_j,a) = \delta_\mathcal{B}(q^{i+1}_j,a) <_\mathcal{B} \delta_\mathcal{B}(q^f_g,a)$, since $\delta_\mathcal{B}(q^f_g,a)$ is either $q^{f+1}_g$ or $s$.
Also, for some $x,y$, $\delta_\mathcal{B}(q^i_j,b) = q^i_x <_\mathcal{B} q^f_y = \delta_\mathcal{B}(q^f_g)$, since $i < f$.
\item If $i = f$, then $v_j <_\mathcal{A} v_g$ by definition. If $i=k$ then $\delta_\mathcal{B}(q^i_j,a) = \delta_\mathcal{B}(q^i_g,a) = s$; otherwise $\delta_\mathcal{B}(q^i_j,a) = q^{i+1}_j \le_\mathcal{B} q^{i+1}_g = \delta_\mathcal{B}(q^i_g,a)$ from $v_j <_\mathcal{A} v_g$.
Also, $v_j <_\mathcal{A} v_g$ implies $\delta(v_j,a_i) = v_x \le_\mathcal{A} v_y = \delta(v_g,a_i)$ for some $x,y$. So $\delta_\mathcal{B}(q^i_j,b) = q^i_x \le_\mathcal{B} q^i_y = \delta_\mathcal{B}(q^i_g,b)$.
\end{itemize}
Thus $\le_\mathcal{B}$ is an underlying linear order of $\mathcal{B}_\mathcal{A}$.

Now, suppose that $\mathcal{B}_\mathcal{A}$ is monotonic with an underlying linear order $\le_\mathcal{B}$. We define $\le_\mathcal{A}$ on the states of $\mathcal{A}$: for $1 \le j,g \le n$, $v_j <_\mathcal{A} v_g$ if and only if $q^1_j <_\mathcal{B} q^1_g$.
Observe that for any $j \neq g$, $q^1_j <_\mathcal{B} q^1_g$ implies $q^i_j <_\mathcal{B} q^i_g$ for each $2 \le i \le k$ due to the action of $a$.
Consider two states $v_j, v_g$ with $v_j < v_g$ and the $i$-th letter $a_i$. By definition $q^1_j <_\mathcal{B} q^1_g$, and so $q^i_j <_\mathcal{B} q^i_g$. This implies $\delta_\mathcal{B}(q^i_j,b) = q^i_x \le_\mathcal{B} q^i_y = \delta_\mathcal{B}(q^i_g,b)$ for some $x,y$, and it follows that $q^1_x \le_\mathcal{B} q^1_y$. Thus $\delta(v_j,a_i) = v_x \le_\mathcal{A} v_y = \delta(v_g,a_i)$, and the order $\le_\mathcal{A}$ is an underlying linear order of $\mathcal{A}$.
\qed
\end{proof}

As a corollary we obtain that $\textsc{MONOTONIC}_2$ is also NP-complete. We can reduce an instance of $\textsc{MONOTONE-NAE-3SAT}$ with $n$ variables and $m$ clauses to a binary automaton with $3m(2n+3m+1)+1$ states.

\begin{corollary}
The problem of checking whether a given binary automaton $\mathcal{A}$ is monotonic is NP-complete.
\end{corollary}


\section{Oriented Automata}\label{sec:oriented}

The following definition of oriented automata is due to Eppstein \cite{Ep1990} (who used the term \emph{monotonic}). An automaton is \emph{oriented} if there is a cyclic order of the states preserved by the action of the letters. Formally, there is a cyclic order $q_1,\ldots,q_n$ such that for every $a \in \Sigma$, the sequence $\delta(q_1,a),\ldots,\delta(q_n,a)$, after removal of possibly adjacent duplicate states (the last is also adjacent with the first), is a subsequence of a cyclic permutation $q_i,\ldots,q_n,q_1,\ldots,q_{i-1}$ of the cyclic order, for some $1 \le i \le n$.
Note that if $q_1,\ldots,q_n$ is a cyclic order then also $q_i,\ldots,q_n,q_1,\ldots,q_{i-1}$ is for every $1 \le i \le n$.
Figure~\ref{fig:cyclic_order} presents a cyclic order of some unary oriented automaton.
Every monotonic automaton is oriented, since if a linear order is preserved, then it is also preserved as a cyclic order. But the converse does not necessarily hold.

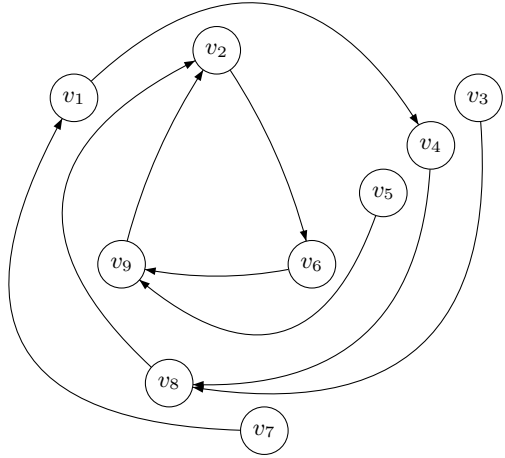
\begin{figure}[ht]
\unitlength 9pt
\begin{center}\begin{picture}(24,17)(0,0.5)
\gasset{Nh=2,Nw=2,Nmr=1,ELdist=0.5,loopdiam=1.5}
\node(v1)(4,14){$v_1$}
\node(v2)(10,16){$v_2$}
\node(v3)(21,14){$v_3$}
\node(v4)(19,12){$v_4$}
\node(v5)(17,10){$v_5$}
\node(v6)(14,7){$v_6$}
\node(v7)(12,0){$v_7$}
\node(v8)(8,2){$v_8$}
\node(v9)(6,7){$v_9$}
\drawedge[curvedepth=0.5](v2,v6){}
\drawedge[curvedepth=0.5](v6,v9){}
\drawedge[curvedepth=0.5](v9,v2){}
\drawedge[curvedepth=5.5](v8,v2){}
\drawedge[curvedepth=4](v4,v8){}
\drawedge[curvedepth=6](v3,v8){}
\drawedge[curvedepth=4.5](v5,v9){}
\drawedge[curvedepth=5](v1,v4){}
\drawedge[curvedepth=7](v7,v1){}
\end{picture}\end{center}
\caption{The cyclic order $(v_1,v_2,\ldots,v_9)$ (clockwise) of a unary automaton.}
\label{fig:cyclic_order}
\end{figure}

Let $\textsc{ORIENTED}$ be the problem of deciding if a given automaton is oriented. As before, we consider $\textsc{ORIENTED}_k$ with the restriction to $k$-letter alphabets. Again, $\textsc{ORIENTED}_1$ can be easily solved in $\O(n)$ time due to the following

\begin{proposition}\label{pro:unary_oriented}
A unary automaton is oriented if and only if all cycles in the transformation of the single letter have the same length.
There is an algorithm solving the problem $\textsc{ORIENTED}_1$ and finding a cyclic order if it exists, and working in $\O(|Q|)$ time.
\end{proposition}
\begin{proof}
Let $a$ be the single letter of the alphabet and $n=|Q|$. Suppose that $(c_1,\ldots,c_k)$ and $(d_1,\ldots,d_\ell)$ are two cycles in the transformation of $a$, with $1 \le k < \ell$. Then, the transformation of $a^k$ has the cycle $(d_1,d_k,\ldots,d_{(m-1)k \mod \ell})$ of length $m$, for some $2 \le m \le \ell$. On the other hand it has the fixed point $c_1$. Let $q_1,\ldots,q_{n-1},q_n=c_1$ be a cyclic order of the states of the automaton. Since the transformation of $a^k$ has a cycle of length $\ge 2$ (which does not involve $c_1$), there are two states $q_i,q_j$ with $i<j<n$, $\delta(q_i,a^k)=q_f$, and $\delta(q_j,a^k)=q_g$, such that $f>g$. It follows that $\delta(q_i,c_1)=q_f,\delta(q_j,c_1)=q_g,\delta(c_1,a^k)=c_1$ violates the cyclic order $q_1,\ldots,q_g,\ldots,q_f,\ldots,q_n=c_1$, since $(q_f,q_g,c_1)$ is not a subsequence of any cyclic permutation of the cyclic order---a contradiction.

Assume now that we have $m$ cycles of the same length $k$:
$$(c^1_1,\ldots,c^1_k),(c^2_1,\ldots,c^2_k),\cdots,(c^m_1,\ldots,c^m_k),$$
so $c^j_i$ is the $i$-th state of the $j$-th cycle, and $\delta(c^j_i) = c^j_{i \mod k + 1}$.
We can compute a cyclic order by breadth-first search (BFS) in the inverse digraph of the transformation of $a$. The constructed cyclic order will have the form
$$Q^1_1,c^1_1,\ldots,Q^m_1,c^m_1,\cdots,Q^1_k,c^1_k,\ldots,Q^m_k,c^m_k,$$
where $Q^j_i$ are sequences of states that do not lie on a cycle. Let $\ell(q)$ (\emph{level}) be the smallest integer $i$ such that $\delta(q,a^i)$ is a state on a cycle. To simplify the notation, let $i \oplus j$ be $(i-1+j) \mod k + 1$.

The algorithm starts from the list $(c^1_1,\ldots,c^m_1,\cdots,c^1_k,\ldots,c^m_k)$ of all cycle states; they are considered as visited in the $0$-th step in this order. In the $i$-th step ($i \ge 1$), the algorithm processes the list of visited states from the $(i-1)$-th step in the order in which they were visited. For each state $p$ from the list, the algorithm computes all states $q$ such that $\delta(q,a)=p$ and $q$ is not a cycle state; so it visits precisely all the states $q$ with $\ell(q)=i$. For every visited state $q$, it appends $q$ to the end of the new list of visited states in the current step. For a visited $q$, we have the corresponding cycle state $c^f_g = \delta(q,a^{\ell(q)})$, from which $q$ was reached (possibly indirectly). The algorithm appends $q$ to the beginning of $Q^f_j$ with $j=(nk+g-1-i) \oplus 1$; for example, if $g=1$, then for $i=1,2,3,\ldots,k-1,k,k+1,\ldots$ we have $j=k,k-1,\ldots,2,1,k,\ldots$, respectively.

To illustrate the algorithm, consider the automaton from~Figure~\ref{fig:cyclic_order}. We start from the list $(c_1,c_2,c_3) = (v_2,v_6,v_9)$ of the one cycle, and empty $Q_1,Q_2,Q_3$. In the first step, from state $v_2$ we reach $v_8$, from $v_6$ we do not reach any state, and from $v_9$ we reach $v_5$. Hence, $Q_3=(v_8)$ as $v_2=c_1$, and $Q_2=(v_5)$ as $v_9=c_3$. Then, in the second step, from $v_8$ we reach $v_4$ and $v_3$, and from $v_5$ we do not reach any state; hence, we append $v_4$ and $v_3$ to the beginning of $Q_2$, obtaining $Q_2=(v_3,v_4,v_5)$. In the third step, from $v_4$ we reach $v_1$, so $Q_1$ becomes $(v_1)$. Finally, in the last fourth step, from $v_1$ we reach $v_7$, obtaining $Q_3=(v_7,v_8)$. The final order is so
$$Q_1,c_1,Q_2,c_2,Q_3,c_3 = v_1,v_2,v_3,v_4,v_5,v_6,v_7,v_8,v_9.$$

We can show that the resulted cyclic order is indeed preserved by the action of $a$. Observe that $\delta(c^j_i,a) = c^j_{i \oplus 1}$, and if $q \in Q^j_i$ then $\delta(q,a) \in Q^j_{i \oplus 1}$ or $\delta(q,a)=c^j_{i \oplus 1}$. Hence, the sequence $Q^j_i,c^j_i$ is mapped into $Q^j_{i \oplus 1},c^j_{i \oplus 1}$, and it remains to show that for each $Q^j_i=(p_1,\ldots,p_s)$, the sequence $\delta(p_1,a),\ldots,\delta(p_s,a),\delta(c^j_i,a)$ is a subsequence of $Q^j_{i \oplus 1},c^j_{i \oplus 1}$.
Consider $<$ as the order in these sequences, and let $u,v$ be two states from $Q^j_i \cup \{c^j_i\}$ with $u < v$. If $v=c^j_i$ then we have $\delta(u,a) \le \delta(v,a)=c^j_{i \oplus 1},a$. If $u=p_f$, $v=p_g$ then $u < v$ means that the algorithm appended $u$ after $v$, so $u$ was visited after $v$. They were directly reached from $\delta(u,a)$ and $\delta(v,a)$, respectively. If $\delta(v,a)=c^j_i$ then $\delta(u,a) \le \delta(v,a)$ clearly holds, and if $\delta(u,a)=c^j_i$ then also $\delta(v,a)=c^j_i$. Otherwise, $\delta(u,a),\delta(v,a) \in Q^j_{i \oplus 1}$ and it follows that $\delta(u,a)$ was visited after $\delta(v,a)$ by the algorithm, so $\delta(u,a) > \delta(v,a)$.
As usual breadth-first search, this procedure works in $\O(|Q|)$ time.
\qed
\end{proof}

To show hardness, we reduce the NP-complete problems $\textsc{MONOTONIC}$ and $\textsc{MONOTONIC}_k$ (with $k \ge 2$) to $\textsc{ORIENTED}$ and $\textsc{ORIENTED}_k$, respectively.

\begin{proposition}
Let $\mathcal{A}_{+1}$ be an automaton obtained from $\mathcal{A}=\langle Q,\Sigma,\delta \rangle$ by adding a unique state $s$ with $\delta(s,a)=s$ for every $a \in \Sigma$. Then the following are equivalent:
\begin{itemize}
\item $\mathcal{A}$ is monotonic;
\item $\mathcal{A}_{+1}$ is monotonic;
\item $\mathcal{A}_{+1}$ is oriented.
\end{itemize}
\end{proposition}
\begin{proof}
Clearly $\mathcal{A}_{+1}$ is monotonic if and only if $\mathcal{A}$ is monotonic, and if $\mathcal{A}_{+1}$ is monotonic then it is also oriented. It remains to show that if $\mathcal{A}_{+1}$ is oriented then $\mathcal{A}_{+1}$ is monotonic.

Assume that $\mathcal{A}_{+1}$ is not monotonic but is oriented, and let $q_1,\ldots,q_n,s$ be a preserved cyclic order of the states of $\mathcal{A}$. Since no state is mapped to $s$, except $s$, and $s$ is mapped to itself under the action of every letter, $q_1,\ldots,q_n$ is a preserved cyclic order of the states of $\mathcal{A}$.
Since $\mathcal{A}$ is not monotonic, $q_1,\ldots,q_n$ is not an underlying linear order of $\mathcal{A}$. So there are two states $q_i,q_j \in Q$ and $a \in \Sigma$, with $i<j$, $\delta(q_i,a) = q_f$, and $\delta(q_j,a) = q_g$, such that $f>g$. It follows that $\delta(q_i,a)=q_f,\delta(q_j,a)=q_g,\delta(s,a)=s$ violates the cyclic order $q_1,\ldots,q_g,\ldots,q_f,\ldots,q_n,s$ of the states of $\mathcal{A}$, since $(q_f,q_g,s)$ is not a subsequence of any cyclic permutation of the cyclic order. Thus $\mathcal{A}_{+1}$ cannot be oriented and not monotonic.
\qed
\end{proof}

\begin{corollary}
The problem of checking whether a given automaton is oriented is NP-complete, even under the restriction to binary alphabets.
\end{corollary}


\section{Discussion}\label{sec:discussion}

We have proved that checking whether an automaton is monotonic or oriented is NP-complete.
However, several related problems remain open.
The complexity of determining whether an automaton is \emph{generalized monotonic} \cite{AV2005SynchronizingGeneralizedMonotonicAutomata}, and \emph{weakly monotonic} \cite{Volkov2009ChainOfPartialOrders} is not known. The class of generalized monotonic automata strictly contains the class of monotonic ones, and the class of weakly monotonic automata strictly contains the class of generalized monotonic ones.
Also, it remains open what is the complexity of checking whether an automaton \emph{respects intervals of a directed graph} \cite{GK2013AutomataRespectingIntervals}; this is the widest of the classes containing the classes of generalized and weakly monotonic automata.

It can be observed that if the alphabet is unary then the classes of generalized and weakly monotonic automata are precisely the class of monotonic automata. However, it is not difficult to check that automata $\mathcal{A}_I$ from the construction from Subsection~\ref{subsec:monotonic_npc} are generalized, and so weakly monotonic, regardless of the instance $I$; thus our proof of NP-completeness of testing monotonicity does not work for these wider classes.

On the other hand, for the class of automata preserving a nontrivial \emph{partial order}, the membership problem can be easily solved in polynomial time. An automaton preserves a partial order $\le$, if $p \le q$ implies $\delta(p,a) \le \delta(q,a)$ for every $p,q \in Q$, $a \in \Sigma$. A partial order is nontrivial if at least one pair of states is comparable. In contrast to monotonic automata, not all pairs of states must be comparable, but at least one.
This class contains monotonic, generalized monotonic, and weakly monotonic automata, but not oriented, and is a subclass of automata respecting intervals of a directed graph. From~\cite{GK2013AutomataRespectingIntervals} it follows that if the \v{C}ern\'y conjecture is true for all automata outside this class (admitting only trivial partial orders), then it is true for all automata.

\begin{proposition}\label{pro:partial_order}
Checking whether an automaton preserves a nontrivial partial order and finding it if exists can be done in $\O(|\Sigma| \times |Q|^6)$ time and $\O(|Q|^2)$ working space.
\end{proposition}
\begin{proof}
For each pair of distinct states $p,q \in Q$, we try to construct a partial order $<$ with $p < q$. So at the beginning of constructing, all states are incomparable and we order $p < q$.
When ordering a pair $x,y \in Q$ with $x < y$, we take all the consequences $\delta(x,a) < \delta(y,a)$ for every $a \in \Sigma$ with $\delta(x,a) \neq \delta(y,a)$. Of course, this also involves that $x' < y'$ for every $x' < \delta(x,a)$ and $y' > \delta(y,a)$. For each newly ordered pair we repeat the procedure of taking consequences. If a contradiction is found, that is, if we need to order $x < y$ but they have been already ordered so that $x > y$, the construction fails and we start from another pair $p,q$. If for some pair $p,q$ all the consequences are taken without a contradiction, we have found a preserved partial order with $p < q$.

Clearly, if the algorithm finds a partial order, then $x \le y$ implies $\delta(x,a) \le \delta(y,a)$ as it has taken all the consequences, so the order is preserved. Conversely, if there exists a preserved nontrivial partial order $\le$, then $p < q$ for some pair of states, and the consequences cannot lead to a contradiction. Hence, the algorithm will find the minimal partial order with $p < q$ that is preserved and is contained in $\le$.

Concerning the complexity, we need to process $\O(|Q|^2)$ pairs. The constructed partial order can be simply stored as a directed acyclic graph. For every $p,q$, we start from the empty digraph with one edge $(p,q)$. For each ordered pair $\{x,y\}$ we need to take or check $\O(|\Sigma|)$ consequences, and we order $\O(|Q|^2)$ pairs. Taking a consequence and updating the constructed partial order takes $\O(|Q|^2)$ time, due to the possibly quadratic size of $\{z \in Q \mid z < x\} \times \{z \in Q \mid z > y\}$. These together yield in $\O(|\Sigma| \times |Q|^6)$ time, and the need of storing digraphs yields in $\O(|Q|^2)$ space.
\qed
\end{proof}

The algorithm from~Proposition~\ref{pro:partial_order} may be modified for finding an underlying linear order of the given automaton. To do so, after finding a partial order that is not yet linear, we need to order another pair that is not yet comparable, say $\{x,y\}$. Here we must consider both possibilities $x < y$ and $x > y$ to check if one of them finally leads to a linear order. Hence, this results in super-exponential worst case running time.
However, based on some of our experimental evidence, this algorithm is practically much more efficient than the naive checking of all linear orderings:
in most cases of not monotonic automata we can find a contradiction quickly, without the need to enumerate directly all orderings.

\bibliographystyle{plain}

\begin{thebibliography}{10}

\bibitem{AV2003SynchronizingMonotonicAutomata}
D.~S. Ananichev and M.~V. Volkov.
\newblock Synchronizing monotonic automata.
\newblock In {\em Developments in Language Theory}, volume 2710 of {\em LNCS},
  pages 111--121. Springer, 2003.

\bibitem{AV2005SynchronizingGeneralizedMonotonicAutomata}
D.~S. Ananichev and M.~V. Volkov.
\newblock {Synchronizing generalized monotonic automata}.
\newblock {\em Theoretical Computer Science}, 330(1):3--13, 2005.

\bibitem{Berlinkov2010Approximating}
M.~V. Berlinkov.
\newblock {Approximating the minimum length of synchronizing words is hard}.
\newblock In {\em Computer Science -- Theory and Applications}, volume 6072 of
  {\em LNCS}, pages 37--47. Springer, 2010.

\bibitem{BrzozowskiKnast1978TheDotDepthHierarchyIsInfinite}
J.~Brzozowski and R.~Knast.
\newblock {The dot-depth hierarchy of star-free languages is infinite}.
\newblock {\em Journal of Computer and System Sciences}, 16(1):37--55, 1978.

\bibitem{BLL13SyntacticComplexitiesOfSixClasses}
J.~Brzozowski, B.~Li, and D.~Liu.
\newblock {Syntactic complexities of six classes of star-free languages}.
\newblock {\em Journal Automata, Languages and Combinatorics}, 17(2-4):83--105,
  2012.

\bibitem{BSX2011DecisionProblemsForConvexLanguages}
J.~Brzozowski, J.~Shallit, and Z.~Xu.
\newblock {Decision problems for convex languages}.
\newblock {\em Information and Computation}, 209(3):353--367, 2011.

\bibitem{BS2014LargeAperiodicSemigroups}
J.~Brzozowski and M.~Szyku{\l}a.
\newblock {Large Aperiodic Semigroups}.
\newblock In {\em Implementation and Application of Automata}, volume 8587 of
  {\em LNCS}, pages 124--135. Springer, 2014.

\bibitem{ChoHuynh1991}
S.~Cho and D.~T. Huynh.
\newblock {Finite-automaton aperiodicity is PSPACE-complete}.
\newblock {\em Theoretical Computer Science}, 88(1):99--116, 1991.

\bibitem{Ep1990}
D.~Eppstein.
\newblock {Reset sequences for monotonic automata}.
\newblock {\em SIAM Journal on Computing}, 19:500--510, 1990.

\bibitem{GH2011}
M.~Gerbush and B.~Heeringa.
\newblock {Approximating minimum reset sequences}.
\newblock In {\em Implementation and Application of Automata}, volume 6482 of
  {\em LNCS}, pages 154--162. Springer, 2011.

\bibitem{GomesHowie1992}
G.~Gomes and J.~Howie.
\newblock On the ranks of certain semigroups of order-preserving
  transformations.
\newblock {\em Semigroup Forum}, 45:272--282, 1992.

\bibitem{GK2013AutomataRespectingIntervals}
M.~Grech and A.~Kisielewicz.
\newblock {The {\v{C}ern\'{y}} conjecture for automata respecting intervals of
  a directed graph}.
\newblock {\em Discrete Mathematics and Theoretical Computer Science},
  15(3):61--72, 2013.

\bibitem{HolzerKutrib2011Survey}
M.~Holzer and M.~Kutrib.
\newblock {Descriptional and computational complexity of finite automata -- A
  survey}.
\newblock {\em Information and Computation}, 209(3):456--470, 2011.

\bibitem{IvNG2014}
S.~Iv\'an and J.~Nagy-Gy\"orgy.
\newblock {On nonpermutational transformation semigroups with an application to
  syntactic complexity}.
\newblock \url{http://arxiv.org/abs/1402.7289}, 2014.

\bibitem{KMM1991APolynomialTimeAlgorithmForLocalTestabilityProblem}
S.~M. Kim, R.~McNaughton, and R.~McCloskey.
\newblock {A polynomial time algorithm for the local testability problem of
  deterministic finite automata}.
\newblock {\em IEEE Transactions on Computers}, 40(10):1087--1093, 1991.

\bibitem{KKS2015ComputingTheShortestResetWords}
A.~Kisielewicz, J.~Kowalski, and M.~Szyku{\l}a.
\newblock {Computing the shortest reset words of synchronizing automata}.
\newblock {\em Journal of Combinatorial Optimization}, 29(1):88--124, 2015.

\bibitem{KS2013GeneratingSmallAutomata}
A.~Kisielewicz and M.~Szyku{\l}a.
\newblock {Generating Small Automata and the \v{C}ern\'{y} Conjecture}.
\newblock In {\em Implementation and Application of Automata}, volume 7982 of
  {\em LNCS}, pages 340--348. Springer, 2013.

\bibitem{KS2014SynchronizingAutomataWithLargeResetLengths}
A.~Kisielewicz and M.~Szyku{\l}a.
\newblock {Synchronizing Automata with Large Reset Lengths}.
\newblock \url{http://arxiv.org/abs/1404.3311}, 2014.

\bibitem{KRW2012}
R.~Kud{\l}acik, A.~Roman, and H.~Wagner.
\newblock {Effective synchronizing algorithms}.
\newblock {\em Expert Systems with Applications}, 39(14):11746--11757, 2012.

\bibitem{Ma2009}
P.~V. Martyugin.
\newblock {Complexity of problems concerning reset words for some partial cases
  of automata}.
\newblock {\em Acta Cybernetica}, 19:517--536, 2009.

\bibitem{MP1971CounterFreeAutomata}
R.~McNaughton and S.~A. Papert.
\newblock {\em {Counter-Free Automata}}, volume~65 of {\em MIT Research
  Monographs}.
\newblock The MIT Press, 1971.

\bibitem{Na1986}
B.~K. Natarajan.
\newblock {An algorithmic approach to the automated design of parts orienters}.
\newblock In {\em Foundations of Computer Science, 27th Annual Symposium on},
  pages 132--142, 1986.

\bibitem{OM2010}
J.~Olschewski and M.~Ummels.
\newblock {The complexity of finding reset words in finite automata}.
\newblock In {\em Mathematical Foundations of Computer Science}, volume 6281 of
  {\em LNCS}, pages 568--579. Springer, 2010.

\bibitem{RSP1993}
J.-K. Rho, Somenzi F., and C.~Pixley.
\newblock Minimum length synchronizing sequences of finite state machine.
\newblock In {\em Proceedings of the 30th ACM/IEEE Design Automation
  Conference}, DAC '93, pages 463--468, 1993.

\bibitem{Sandberg2005Survey}
S.~Sandberg.
\newblock {Homing and synchronizing sequences}.
\newblock In {\em Model-Based Testing of Reactive Systems}, volume 3472 of {\em
  LNCS}, pages 5--33. Springer, 2005.

\bibitem{Schaefer1978}
T.~J. Schaefer.
\newblock {The Complexity of Satisfiability Problems}.
\newblock In {\em Proceedings of the Tenth Annual ACM Symposium on Theory of
  Computing}, STOC, pages 216--226. ACM, 1978.

\bibitem{ST2011}
E.~Skvortsov and E.~Tipikin.
\newblock {Experimental study of the shortest reset word of random automata}.
\newblock In {\em Implementation and Application of Automata}, volume 6807 of
  {\em LNCS}, pages 290--298. Springer, 2011.

\bibitem{Volkov2009ChainOfPartialOrders}
M.~V. Volkov.
\newblock Synchronizing automata preserving a chain of partial orders.
\newblock {\em Theoretical Computer Science}, 410(37):3513--3519, 2009.

\end{thebibliography}

\end{document}